\documentclass[11pt]{article}
\usepackage[utf8]{inputenc}
\usepackage{amsmath,amssymb,amsthm}
\usepackage{algorithm}
\usepackage{algpseudocode}
\usepackage{hyperref}
\usepackage{geometry}
\usepackage{graphicx}
\newtheorem{theorem}{Theorem}
\newtheorem{lemma}{Lemma}
\newtheorem{definition}{Definition}
\newtheorem{remark}{Remark}

\title{On the Complexity of the Ordered Covering Problem in Distance Geometry}

\author{
Michael Souza$^{1}$ \and J\'ulio Ara\'ujo$^{1}$ \and John Kesley Costa$^{1}$ \and Carlile Lavor$^{2}$\\
\\
$^{1}$ Federal University of Cear\'a, Fortaleza, Brazil\\
$^{2}$ University of Campinas (IMECC-UNICAMP), Campinas, Brazil
}

\date{}

\begin{document}

\maketitle


\begin{abstract}
The Ordered Covering Problem (OCP) arises in the context of the Discretizable Molecular Distance Geometry Problem (DMDGP), where the ordering of pruning edges significantly impacts the performance of the SBBU algorithm for protein structure determination. In recent work, Souza et al. (2023) formalized OCP as a hypergraph covering problem with ordered, exponential costs, and proposed a greedy heuristic that outperforms the original SBBU ordering by orders of magnitude. However, the computational complexity of finding optimal solutions remained open. In this paper, we prove that OCP is NP-complete through a polynomial-time reduction from the strongly NP-complete 3-Partition problem. Our reduction constructs a tight budget that forces optimal solutions to correspond exactly to valid 3-partitions. This result establishes a computational barrier for optimal edge ordering and provides theoretical justification for the heuristic approaches currently used in practice.
\end{abstract}

\section{Introduction}

The Distance Geometry Problem (DGP) is fundamental to computational structural biology: given a weighted graph $G=(V,E,d)$ with distance function $d: E \to \mathbb{R}_+$ and a positive integer $K$, find a realization $x: V \to \mathbb{R}^K$ such that $\|x_u - x_v\| = d_{uv}$ for all $\{u,v\} \in E$ \cite{liberti2017euclidean,liberti2014euclidean,crippen1988distance}. Saxe~\cite{saxe1979embeddability} proved that the DGP is strongly NP-hard, even for fixed dimension $K$. This problem has extensive applications in determining three-dimensional protein structures from experimental distance data obtained via Nuclear Magnetic Resonance (NMR) spectroscopy \cite{wuthrich1989protein}.

For arbitrary graphs, the DGP typically requires search algorithms in continuous space \cite{liberti2014euclidean}. However, certain protein-derived graphs possess special structural properties that enable discrete solution methods. In particular, protein backbone graphs exhibit a natural ordering of atoms along the molecular chain, where each atom's position can be determined from a finite set of candidate positions based on its predecessors \cite{lavor2019minimal}. This discretization property led to the development of the Discretizable Molecular Distance Geometry Problem (DMDGP) \cite{lavor2012discretizable}, a subclass of DGP instances characterized by vertex orders satisfying specific adjacency constraints.

The DMDGP framework has proven remarkably effective for protein backbone determination. When the graph has an appropriate vertex order, the solution space becomes discrete, and a Branch-and-Prune (BP) algorithm \cite{liberti2008branch} can systematically explore all feasible molecular conformations. The BP algorithm constructs solutions incrementally, branching on the possible positions of each atom and pruning configurations that violate distance constraints. While exponential in the worst case, BP remains practical for proteins due to effective pruning from additional distance constraints.

Recently, Gon\c{c}alves et al.\ \cite{goncalves2021new} introduced the SBBU algorithm, which represents a significant advance over the classical BP approach. Rather than exploring a single global search tree, SBBU decomposes the DMDGP into a sequence of smaller subproblems, each associated with a pruning edge (a distance constraint beyond those required for discretization). The algorithm solves these subproblems sequentially, and crucially, binary variables from solved subproblems can be eliminated from subsequent ones through symmetry exploitation \cite{mucherino2012exploiting}.

The computational efficiency of SBBU depends critically on the order in which pruning edges are processed. Given a permutation $\pi = (e_1, \ldots, e_m)$ of the $m$ pruning edges, SBBU solves the sequence $(P(e_1), \ldots, P(e_m))$ of feasibility subproblems. Each subproblem $P(e_i)$ involves a set of binary variables whose cardinality determines an exponential search space. The ordering determines which variables can be eliminated early, dramatically affecting the total computational cost.

Souza et al.\ \cite{souza2023ordered} demonstrated this impact empirically: on a test set of 5,000 randomly generated instances, the original SBBU edge ordering performed on average 1,300 times worse than optimal orderings, with some instances showing gaps exceeding 6,000-fold. They formalized the edge ordering problem as the Ordered Covering Problem (OCP), viewing the pruning edges as hyperedges covering segments of binary variables, where the order determines which variables remain active in each subproblem. Their proposed greedy heuristic achieved near-optimal performance (within 0.1\% on average), but the fundamental computational complexity of finding optimal solutions remained an open question.

The complexity of vertex ordering problems in distance geometry has been extensively studied. Cassioli et al.\ \cite{cassioli2015discretization} analyzed three types of discretization vertex orders, establishing their NP-completeness and inclusion relationships:

\begin{enumerate}
\item The \textbf{Trilateration Ordering Problem (TOP)} asks whether a graph admits a DDGP order, where each vertex beyond an initial $K$-clique is adjacent to at least $K$ predecessors. Cassioli et al.\ proved TOP is NP-complete by reduction from Clique. However, for fixed $K$, TOP becomes polynomial-time solvable \cite{lavor2012discretization}.
\item The \textbf{Contiguous Trilateration Ordering Problem (CTOP)} requires the $K$ adjacent predecessors to be contiguous in the order. This defines the stricter ${}^K$DMDGP class, which has favorable symmetry properties \cite{lavor2012discretizable}. Cassioli et al.\ showed CTOP is NP-complete for any fixed $K$ by reduction from Hamiltonian Path. MacNeil and Bodur \cite{macneil2022constraint} later developed constraint programming approaches that outperform integer programming formulations for CTOP.
\item The \textbf{Re-Order Problem (ReOP)} allows vertex repetitions in the order, enabling discretization of more instances. Cassioli et al.\ established ReOP is NP-complete for any fixed $K$ by restriction from CTOP.
\end{enumerate}

These results establish the inclusionwise relationship: CTOP $\subsetneq$ ReOP $\subsetneq$ TOP, with all three problems being NP-complete. More recently, Lavor et al.\ \cite{lavor2019polynomiality} showed an interesting dichotomy: the ReOP is NP-complete for $K=1$ (essentially Hamiltonian Path), but belongs to $\mathbf{P}$ for any fixed $K \geq 2$, with a polynomial-time algorithm running in $O(|V|^{2K})$ time.

The OCP differs fundamentally from these vertex ordering problems. While TOP, CTOP, and ReOP concern the order of graph vertices with local adjacency constraints, OCP concerns the order of pruning edges (or more generally, hyperedges in a hypergraph) with a global, exponential cost function. The cost structure arises from the exponential growth of search spaces with respect to the number of active binary variables.

Formally, an OCP instance consists of a set $S$ of labels with weights, a family $\mathcal{E}$ of subsets covering $S$, and a budget $C$. An ordered covering is a sequence $\mathcal{E}' = (E'_1, \ldots, E'_k)$ of sets from $\mathcal{E}$ that collectively cover $S$. The cost is defined by residual sets: $U_i = E'_i \setminus \bigcup_{j<i} E'_j$ represents elements first covered by $E'_i$, and the cost contribution is $f(E'_i) = 2^{u_i}$ where $u_i = \sum_{x \in U_i} \text{val}(x)$. The total cost is $F(\mathcal{E}') = \sum_i f(E'_i)$. The decision problem asks: does there exist an ordered covering with $F(\mathcal{E}') \leq C$? Unless stated otherwise, all integers in the input (weights and budget) are encoded in binary.

The exponential cost function fundamentally distinguishes OCP from classical covering problems like Set Cover or Hitting Set, which have linear or polynomial costs \cite{karp1972reducibility,garey1979computers,chvatal1979greedy}. Moreover, the order dependence means that simply finding a minimal cover is insufficient---the sequence matters critically. Taken together, these structural features might suggest that OCP is computationally tractable, perhaps even admitting a greedy algorithm or a dynamic programming approach.

In this paper, we resolve the complexity of OCP by proving it is NP-complete. Our main result is:

\begin{theorem}
The Ordered Covering Problem is NP-complete.
\end{theorem}

We establish this through a polynomial-time reduction from the 3-Partition problem, which is known to be strongly NP-complete \cite{garey1979computers}. The reduction uses a carefully calibrated exponential cost structure to separate valid from invalid solutions, imposes a tight budget that is met exactly when a 3-partition exists, and builds opening and assignment edges that force any feasible solution to respect the triplet structure of the 3-partition instance.

The proof relies on three key lemmas establishing that any feasible solution must have a rigid structure: the number of active assignment edges equals the number of bins in the corresponding 3-Partition instance, each preceded by its opening edge, each contributing exactly three new labels with disjoint coverage. These structural constraints ensure a bijection between OCP solutions and 3-partitions.

Note that although we reduce from the strongly NP-complete 3-Partition problem, the budget parameter in our construction may be exponential in the size of the 3-Partition instance. Consequently, this reduction establishes NP-completeness of OCP but does not establish strong NP-completeness; we return to this point when describing the reduction in Section~3.

Our NP-completeness result provides theoretical justification for the heuristic approach taken by Souza et al.\ \cite{souza2023ordered}. The greedy heuristic's near-optimal performance is not merely a practical convenience but a well-motivated strategy given the fundamental computational barrier. The result also explains why the original SBBU ordering can be arbitrarily suboptimal: no polynomial-time algorithm can guarantee optimality, assuming P $\neq$ NP.

From a broader perspective, our result adds to the landscape of complexity results in distance geometry. While vertex ordering problems (TOP, CTOP, ReOP) are NP-complete via reductions from Clique and Hamiltonian Path, edge ordering (OCP) is NP-complete via reduction from 3-Partition, highlighting the distinct computational structure of the two problem types.

The remainder of this paper is organized as follows. Section 2 provides preliminaries, including formal definitions of OCP and 3-Partition. Section 3 presents our main result: the NP-completeness proof via reduction from 3-Partition, including the reduction construction and correctness proof with three key lemmas. Section 4 concludes with a discussion of implications and open questions.

\section{Preliminaries}

\subsection{The Ordered Covering Problem}

Souza et al.\ \cite{souza2023ordered} originally define the OCP as a minimization problem. In this work, we study the associated decision problem, which is the standard formulation used to prove NP-completeness. The optimization and decision versions are polynomially equivalent, so NP-completeness of the decision version implies NP-hardness of the original optimization problem.

\begin{definition}[OCP Instance]
An instance of the Ordered Covering Problem consists of:
\begin{enumerate}
\item A finite set $S = \{s_1, \ldots, s_n\}$ of \emph{labels};
\item A family $\mathcal{E} = \{E_1, \ldots, E_m\}$ of subsets with $S \subset \bigcup_{i=1}^m E_i$;
\item A \emph{weight function} $\text{val}: \bigcup_{i=1}^m E_i \to \mathbb{N}$ assigning a positive integer weight to every element that appears in some edge;
\item A \emph{budget} $C \in \mathbb{N}$.
\end{enumerate}
\end{definition}
In the DMDGP context, $S$ represents segments of binary variables, $\mathcal{E}$ represents hyperedges corresponding to pruning edges in the distance geometry graph, and the exponential cost reflects the size of the search space for each subproblem in the SBBU algorithm \cite{souza2023ordered}.

\begin{definition}[Ordered Covering]
An \emph{ordered covering} (or simply \emph{covering}) of $S$ is a tuple $\mathcal{E}' = (E'_1, \ldots, E'_k)$ where each $E'_i \in \mathcal{E}$ and $S \subset \bigcup_{i=1}^k E'_i$.
\end{definition}

\begin{definition}[Residual Sets and Weights]
Given a covering $\mathcal{E}' = (E'_1, \ldots, E'_k)$, we define for each $i \in \{1, \ldots, k\}$:
\begin{enumerate}
\item The \emph{residual set}: $U_i = E'_i \setminus \bigcup_{j=1}^{i-1} E'_j$ (with $U_1 = E'_1$);
\item The \emph{residual weight}: $u_i = \sum_{x \in U_i} \text{val}(x)$.
\end{enumerate}
The residual set $U_i$ represents the labels that are covered for the first time by $E'_i$.
\end{definition}

\begin{definition}[Cost Function]
The \emph{partial cost} of $E'_i$ in the covering $\mathcal{E}'$ is:
$$f(E'_i) = \begin{cases} 2^{u_i} & \text{if } u_i > 0, \\ 0 & \text{if } u_i = 0. \end{cases}$$

The \emph{total cost} of the covering $\mathcal{E}'$ is:
$$F(\mathcal{E}') = \sum_{i=1}^k f(E'_i).$$
\end{definition}

The exponential cost function distinguishes OCP from classical covering problems. In Set Cover, for instance, the cost is simply the number of sets selected (or their cardinality) \cite{johnson1973approximation}. In OCP, the cost grows exponentially with the weight of elements covered, and crucially, the order determines which elements contribute to which residual sets.

\begin{definition}[Ordered Covering Problem - Decision Version]
Given an OCP instance $(S, \text{val}, \mathcal{E}, C)$, the decision problem asks:

\emph{Does there exist an ordered covering $\mathcal{E}'$ with $F(\mathcal{E}') \leq C$?}
\end{definition}

\subsection{The 3-Partition Problem}

Our reduction uses the 3-Partition problem, a classical strongly NP-complete problem. Garey and Johnson~\cite{garey1975complexity} originally proved 3-Partition to be NP-complete by a reduction from 3-dimensional matching; see also~\cite{garey1979computers}. The ``strongly'' NP-complete designation means the problem remains NP-complete even when all numbers are bounded by a polynomial in the input size.

\begin{definition}[3-Partition]
An instance of 3-Partition consists of:
\begin{enumerate}
\item A \emph{multiset} $A = \{a_1, \ldots, a_{3m}\}$ of positive integers;
\item A positive integer $B$ such that: (a) $\sum_{i=1}^{3m} a_i = mB$ (the sum equals $m$ times $B$); (b) $\frac{B}{4} < a_i < \frac{B}{2}$ for all $i \in \{1, \ldots, 3m\}$.
\end{enumerate}
The decision problem asks: \emph{Does there exist a partition $\mathcal{P} = \{P_1, \ldots, P_m\}$ of $A$ into $m$ triplets such that $\sum_{a \in P_i} a = B$ for all $i \in \{1, \ldots, m\}$?}
\end{definition}

The constraint $B/4 < a_i < B/2$ ensures that each triplet in a valid partition contains exactly three distinct elements---no pair sums to $B$, and no element can appear twice or with a fourth element. This ``forced triplet'' structure is essential to our reduction.

Note that $A$ is a \emph{multiset}, meaning that multiple elements can have the same value. However, elements are distinguishable by their labels (indices). When we construct a partition, we partition the elements themselves (by their labels), not merely their values. This distinction will be important in our reduction, where we use labeled tokens to represent multiset elements.

\begin{theorem}[Garey and Johnson \cite{garey1975complexity,garey1979computers}]
3-Partition is strongly NP-complete.
\end{theorem}

We use 3-Partition as the source problem, but note that our OCP construction employs exponential cost terms (powers of two). Thus, while the instance encoding size remains polynomial (binary representation), the reduction does not bound all numeric parameters by a polynomial in $m$; accordingly, our proof establishes NP-completeness, not strong NP-completeness, for OCP.

\subsection{Notation and Conventions}

Throughout this paper, we use the following conventions:
\begin{enumerate}
\item $[m] = \{1, 2, \ldots, m\}$ denotes the set of first $m$ positive integers.
\item For a covering $\mathcal{E}' = (E'_1, \ldots, E'_k)$ and index $i$, we write $U(E'_i)$ or simply $U_i$ for the residual set of $E'_i$ in this covering.
\item We use Greek letters $\alpha, \beta, \gamma$ for labels in $S$, and Latin letters $a, b, c$ for values (weights).
\item We distinguish between \emph{labels} (elements of $S$) and their \emph{weights} (values assigned by the function val).
\end{enumerate}

\section{Main Result: NP-Completeness of OCP}

In this section, we prove the main result of this paper.

\begin{theorem}
The Ordered Covering Problem is NP-complete.
\end{theorem}

The proof consists of two parts: showing OCP is in NP (Section 3.1), and reducing 3-Partition to OCP in polynomial time (Sections 3.2-3.5).

\subsection{OCP is in NP}

\begin{proof}[Proof that OCP $\in$ NP]
We must show that OCP solutions can be verified in time polynomial in the input size when integers are given in binary.

Given an OCP instance $(S, \text{val}, \mathcal{E}, C)$ and a candidate solution $\mathcal{E}' = (E'_1, \ldots, E'_k)$, we verify as follows:

\begin{enumerate}
\item \textbf{Structure of the covering:} We may assume the certificate encodes $\mathcal{E}'$ as a sequence of indices in $[m]$, where $\mathcal{E} = \{E_1,\ldots,E_m\}$. Checking that each index lies in $[m]$ takes $O(k)$ arithmetic comparisons. Using a Boolean array over $S$, we then scan all sets $E'_i$ once to mark covered labels and verify that every element of $S$ is covered at least once. This requires $O\!\left(\sum_{i=1}^k |E'_i|\right)$ operations, which is polynomial in the input size.

\item \textbf{Residual sets and weights:} Processing the sets in order and maintaining a second Boolean array of “already seen” labels over $\bigcup_{i=1}^m E_i$, we compute, for each $i$, the residual set $U_i = E'_i \setminus \bigcup_{j<i} E'_j$ in time $O\!\left(\sum_{i=1}^k |E'_i|\right)$. For each $i$, we then compute $u_i = \sum_{x \in U_i} \text{val}(x)$ using big-integer addition. Let $L$ be the maximum bit-length of any weight in the input. Since $|U_i| \le |E'_i| \le \max_{j} |E_j| \le \sum_{j=1}^m |E_j|$, we have
$$u_i = \sum_{x \in U_i} \text{val}(x) \le |U_i| \times \max_{x \in \bigcup_{j=1}^m E_j}(\text{val}(x)) < \left(\sum_{j=1}^m |E_j|\right) \times 2^L,$$
so each $u_i$ has bit-length at most $L + \left\lceil\log_2 \left(\sum_{j=1}^m |E_j|\right)\right\rceil$, and the total time to compute all $u_i$ is polynomial in the input size.

\item \textbf{Cost computation and budget test:} For each $i$, if $u_i > \lfloor\log_2 C\rfloor$ then $f(E'_i) = 2^{u_i} > C$, so we can immediately reject the certificate. Otherwise, $u_i \le \lfloor\log_2 C\rfloor$ and $f(E'_i) = 2^{u_i}$ has at most $O(\log C)$ bits and can be computed as a left shift. Summing these $k$ values and comparing the result to $C$ uses a polynomial number of big-integer operations on $O(\log C)$-bit integers.
\end{enumerate}

Therefore, verification runs in polynomial time in the input size, and OCP $\in$ NP under binary encoding.
\end{proof}

\subsection{Reduction from 3-Partition to OCP}

We now construct a polynomial-time reduction $\Phi$ that maps any 3-Partition instance to an OCP instance such that the 3-Partition instance is a YES instance if and only if the corresponding OCP instance is a YES instance.

Let $(A, B)$ be an instance of 3-Partition, where $A = \{a_1, \ldots, a_{3m}\}$ with $\sum_{i=1}^{3m} a_i = mB$ and $B/4 < a_i < B/2$ for all $i$. We construct an OCP instance $\Phi(A, B) = (S, \text{val}, \mathcal{E}, C)$ as follows.

\textbf{The Label Set $\boldsymbol{S}$.} The set of \emph{primary labels} corresponding to the elements of the multiset $A$ is:
$$S = \{\alpha_1, \alpha_2, \ldots, \alpha_{3m}\},$$
where each label $\alpha_\ell$ has weight $\text{val}(\alpha_\ell) = a_\ell$.

We use distinct labels $\alpha_\ell$ even when values $a_\ell$ repeat in the multiset $A$. This allows us to distinguish between different elements with the same value, which is essential for partitioning the multiset. These are the only labels that must be covered.

\textbf{Valid Triplets.} Define the collection of \emph{valid triplets}:
$$T = \left\{ X \in \binom{S}{3} : \sum_{\alpha \in X} \text{val}(\alpha) = B \right\}.$$

In other words, $T$ consists of all three-element subsets of $S$ whose weights sum to exactly $B$. By the constraint $B/4 < a_i < B/2$, any valid triplet contains exactly three distinct labels. No pair sums to $B$, and no element appears twice. Finally, we have $|T| \leq \binom{3m}{3} = O(m^3)$, which is polynomial in $m$.

\textbf{Auxiliary Tokens and Edge Construction.} For each ``bin'' $i \in [m]$ and each valid triplet $X_j \in T$ (indexed by $j \in [|T|]$), we introduce two distinct \emph{auxiliary tokens}:
\begin{enumerate}
\item An \emph{opening token} $\omega_{ij}$ with weight $\text{val}(\omega_{ij}) = w$;
\item A \emph{closing token} $\tau_{ij}$ with weight $\text{val}(\tau_{ij}) = t$,
\end{enumerate}

where $w$ and $t$ are positive integer constants to be specified below. Crucially, these tokens are distinct for each pair $(i, j)$ and do not belong to $S$:
$$\omega_{ij}, \tau_{ij} \notin S \quad \text{for all } (i, j).$$
The auxiliary tokens lie outside the required coverage set but still contribute to costs through $\text{val}$.

For each pair $(i, j)$ with $i \in [m]$ and $j \in [|T|]$, we define two edges:
\begin{enumerate}
\item \textbf{Opening edge:} $A_{ij} = \{\omega_{ij}\}$;
\item \textbf{Assignment edge:} $E_{ij} = X_j \cup \{\omega_{ij}, \tau_{ij}\}$.
\end{enumerate}

If $X_j = \{\alpha_{j1}, \alpha_{j2}, \alpha_{j3}\}$, then:
$$E_{ij} = \{\alpha_{j1}, \alpha_{j2}, \alpha_{j3}, \omega_{ij}, \tau_{ij}\}.$$

The edge family is:
$$\mathcal{E} = \{A_{ij}, E_{ij} : i \in [m], j \in [|T|]\}.$$

\textbf{Key properties:} (1) Each $\omega_{ij}$ appears in exactly two edges: $A_{ij}$ and $E_{ij}$; (2) Each $\tau_{ij}$ appears in exactly one edge: $E_{ij}$; (3) Only the $E_{ij}$ edges contain labels from $S$; (4) The total number of edges is $2m|T| = O(m^4)$, which is polynomial in $m$.

The elements that appear in edges (and thus receive weights via $\text{val}$) are $S \cup \{\omega_{ij}, \tau_{ij} : i \in [m], j \in [|T|]\}$, whose cardinality is $3m + 2m|T| = O(m^4)$.

\textbf{The Cost Separation Parameter $\boldsymbol{w}$.} The parameter $w$ is chosen to ensure a cost ``separation'' between different solution structures. Specifically, we want to ensure that using more than $m$ opening edges, or using an assignment edge $E_{ij}$ before its opening edge $A_{ij}$, will exceed the budget.

We first define the \emph{canonical cost} that a valid solution should achieve. If we have a 3-partition $\mathcal{P} = \{P_1, \ldots, P_m\}$, we can construct a covering by selecting, for each bin $i$, the opening edge $A_{ij(i)}$ and assignment edge $E_{ij(i)}$ where $X_{j(i)} = P_i$ (viewing $P_i$ as a subset of labels from $S$). The covering sequence is:
$$(A_{1j(1)}, E_{1j(1)}, A_{2j(2)}, E_{2j(2)}, \ldots, A_{mj(m)}, E_{mj(m)}).$$

In this sequence: (1) Each $A_{ij(i)}$ covers $\omega_{ij(i)}$ with residual weight $w$, contributing cost $2^w$; (2) Each $E_{ij(i)}$ covers $\tau_{ij(i)}$ and the three labels of $P_i$, with residual weight $t + \sum_{\alpha \in P_i} \text{val}(\alpha) = t + B$, contributing cost $2^{t+B}$.

Thus, the total cost of this canonical covering is:
$$F_{\text{can}} = m \cdot 2^w + m \cdot 2^{t+B} = m(2^w + 2^{t+B}).$$

We set the budget $C = F_{\text{can}} = m(2^w + 2^{t+B})$.

Now, we choose $w$ large enough so that exceeding $m$ openings or misusing an assignment edge will violate the budget. Specifically, we require:
$$(m+1) \cdot 2^w > C = m(2^w + 2^{t+B}).$$

Simplifying:
$$(m+1) \cdot 2^w > m \cdot 2^w + m \cdot 2^{t+B},$$
$$2^w > m \cdot 2^{t+B},$$
$$w > t + B + \log_2 m.$$

Thus, we set:
$$w = t + B + \lceil \log_2 m \rceil + 1.$$

The weight $w = O(B + \log m)$ is polynomially bounded. Although the budget $C = m(2^w + 2^{t+B})$ is exponentially large, its bit-length is $O(w) = O(B + \log m)$, which is polynomial in the input size.

We fix $t = 1$ for simplicity (any positive integer works).

\textbf{Summary of the Reduction.} The OCP instance $\Phi(A, B) = (S, \text{val}, \mathcal{E}, C)$ is fully specified:
\begin{enumerate}
\item \textbf{Label set:} $S = \{\alpha_1, \ldots, \alpha_{3m}\}$ with $|S| = 3m$; the union of edge elements is $S \cup \{\omega_{ij}, \tau_{ij} : i \in [m], j \in [|T|]\}$.
\item \textbf{Weight function:} $\text{val}(\alpha_\ell) = a_\ell$, $\text{val}(\omega_{ij}) = w$, $\text{val}(\tau_{ij}) = t = 1$.
\item \textbf{Edge family:} $\mathcal{E} = \{A_{ij}, E_{ij} : i \in [m], j \in [|T|]\}$ with $|\mathcal{E}| = O(m^4)$.
\item \textbf{Budget:} $C = m(2^w + 2^{t+B})$.
\end{enumerate}

The reduction is clearly computable in polynomial time: computing $T$ takes $O(m^3)$ time, constructing edges takes $O(m^4)$ time, and computing $w$ and $C$ takes polynomial time in the bit-lengths of $B$ and $m$.

\subsection{Correctness: Completeness Direction ($\Rightarrow$)}

We first show that if the 3-Partition instance is a YES instance, then the OCP instance is a YES instance.

\begin{theorem}[Completeness]
If $(A, B)$ is a YES instance of 3-Partition, then $\Phi(A, B)$ is a YES instance of OCP, with an optimal covering achieving cost exactly $C$.
\end{theorem}

\begin{proof}
Suppose $\mathcal{P} = \{P_1, \ldots, P_m\}$ is a 3-partition of $A$, where each $P_i = \{a_{i1}, a_{i2}, a_{i3}\}$ (viewing $P_i$ as a set of element labels from $A$) satisfies $\sum_{a \in P_i} a = B$.

For each $i \in [m]$, let $Q_i \subseteq S$ be the set of labels corresponding to $P_i$:
$$Q_i = \{\alpha_\ell : a_\ell \in P_i\}.$$

Since $\mathcal{P}$ is a partition of $A$ by element labels, the sets $Q_1, \ldots, Q_m$ are pairwise disjoint and their union is $S$:
$$Q_i \cap Q_{i'} = \emptyset \text{ for } i \neq i', \quad \bigcup_{i=1}^m Q_i = S.$$

Moreover, $\sum_{\alpha \in Q_i} \text{val}(\alpha) = B$ for each $i$.

For each $i \in [m]$, there exists an index $j(i) \in [|T|]$ such that $X_{j(i)} = Q_i$, because $Q_i$ is a three-element subset of $S$ with total weight $B$, hence $Q_i \in T$.

We construct the covering:
$$\mathcal{E}_{\text{can}} = (A_{1j(1)}, E_{1j(1)}, A_{2j(2)}, E_{2j(2)}, \ldots, A_{mj(m)}, E_{mj(m)}).$$

\textbf{Analysis of costs:}
\begin{enumerate}
\item \textbf{Bin 1:} When we apply $A_{1j(1)} = \{\omega_{1j(1)}\}$, the residual set is $U(A_{1j(1)}) = \{\omega_{1j(1)}\}$ with weight $w$, so $f(A_{1j(1)}) = 2^w$. Next, when we apply $E_{1j(1)} = Q_1 \cup \{\omega_{1j(1)}, \tau_{1j(1)}\}$, the residual set is $U(E_{1j(1)}) = Q_1 \cup \{\tau_{1j(1)}\}$ (since $\omega_{1j(1)}$ was already covered), with weight $\sum_{\alpha \in Q_1} \text{val}(\alpha) + t = B + t$. Thus $f(E_{1j(1)}) = 2^{B+t}$.
\item \textbf{Bin $i$ (general):} Similarly, $f(A_{ij(i)}) = 2^w$ and $f(E_{ij(i)}) = 2^{B+t}$.
\item \textbf{Coverage:} The sets $Q_1, \ldots, Q_m$ are disjoint and cover $S$. Each $\omega_{ij(i)}$ is covered by $A_{ij(i)}$, and each $\tau_{ij(i)}$ is covered by $E_{ij(i)}$. Thus, all labels in $S$ are covered.
\item \textbf{Total cost:} $F(\mathcal{E}_{\text{can}}) = \sum_{i=1}^m (f(A_{ij(i)}) + f(E_{ij(i)})) = \sum_{i=1}^m (2^w + 2^{B+t}) = m(2^w + 2^{B+t}) = C$.
\end{enumerate}

Thus, $\mathcal{E}_{\text{can}}$ is a feasible covering with $F(\mathcal{E}_{\text{can}}) \leq C$.
\end{proof}

\subsection{Correctness: Soundness Direction ($\Leftarrow$)}

We now show that if the OCP instance is a YES instance, then the 3-Partition instance is a YES instance. This is the more involved direction, requiring three key lemmas about the structure of optimal coverings.

\begin{theorem}[Soundness]\label{thm:soundness}
If $\Phi(A, B)$ is a YES instance of OCP (i.e., there exists a covering $\mathcal{E}'$ with $F(\mathcal{E}') \leq C$), then $(A, B)$ is a YES instance of 3-Partition.
\end{theorem}

The proof relies on three lemmas that establish the structure of any feasible covering.

\begin{lemma}[Opening Precedes Assignment]\label{lem:opening}
In any covering $\mathcal{E}' = (E'_1, \ldots, E'_k)$ with $F(\mathcal{E}') \leq C$, every assignment edge $E_{ij}$ that appears with positive residual weight (i.e., $U(E_{ij}) \neq \emptyset$) is preceded in the sequence by its opening edge $A_{ij}$.
\end{lemma}

\begin{proof}
We argue by contradiction. Suppose there is an assignment edge $E_{ij}$ with positive residual such that its opening edge $A_{ij}$ does not appear before $E_{ij}$ in $\mathcal{E}'$.

Since only assignment edges contain labels from $S$, and $|S| = 3m$ while each assignment edge contains at most three labels from $S$, at least $m$ assignment edges must have positive residuals; let $r \geq m$ denote their number.

When $E_{ij}$ is applied without $A_{ij}$ having been applied earlier, both $\omega_{ij}$ and $\tau_{ij}$ are uncovered, so $\omega_{ij}, \tau_{ij} \in U(E_{ij})$ and
$$u(E_{ij}) \geq \text{val}(\omega_{ij}) + \text{val}(\tau_{ij}) = w + t,$$
which implies $f(E_{ij}) \geq 2^{w+t}$.

For each of the remaining $(r-1)$ assignment edges with positive residuals, the covering must pay at least $2^w$: if such an edge is properly opened, its opening edge costs $2^w$, and if it is not, its cost is even larger. Thus,
$$F(\mathcal{E}') \geq (r-1) \cdot 2^w + 2^{w+t}.$$

Since $t \geq 1$, we have $2^{w+t} \geq 2^{w+1} = 2 \cdot 2^w$, so
$$(r-1) \cdot 2^w + 2^{w+t} \geq (r+1) \cdot 2^w \geq (m+1) \cdot 2^w.$$

By our choice of $w$, $(m+1) \cdot 2^w > C$, so $F(\mathcal{E}') > C$, contradicting the assumption that $F(\mathcal{E}') \leq C$.
\end{proof}

\begin{lemma}[Exactly $m$ Active Assignment Edges]\label{lem:exactly-m}
In any covering $\mathcal{E}'$ with $F(\mathcal{E}') \leq C$, exactly $m$ assignment edges appear with positive residuals.
\end{lemma}

\begin{proof}
Let $r$ be the number of assignment edges with positive residuals in $\mathcal{E}'$.

\textbf{Lower bound:} Since only assignment edges contain labels from $S$, and $|S| = 3m$, and each assignment edge contains at most three labels from $S$, we need at least $m$ such edges: $r \geq m$.

\textbf{Upper bound:} By Lemma~\ref{lem:opening}, each of the $r$ assignment edges is preceded by its opening edge. Thus, the covering contains at least $r$ opening edges. The cost contribution from these openings is $r \cdot 2^w$.

The budget is $C = m(2^w + 2^{B+t})$. Since $F(\mathcal{E}') \leq C$, we have:
$$r \cdot 2^w \leq F(\mathcal{E}') \leq C = m(2^w + 2^{B+t}).$$

If $r \geq m+1$, then:
$$r \cdot 2^w \geq (m+1) \cdot 2^w > C,$$
by our choice of $w$, which is a contradiction.

Thus, $r \leq m$. Combined with $r \geq m$, we conclude $r = m$.
\end{proof}

\begin{lemma}[Triplet Coverage of $S$]\label{lem:triplet}
In any covering $\mathcal{E}'$ with $F(\mathcal{E}') \leq C$, the $m$ assignment edges with positive residuals have pairwise disjoint contributions to $S$, each of size three, whose union is exactly $S$.
\end{lemma}

\begin{proof}
By Lemma~\ref{lem:exactly-m}, there are exactly $m$ assignment edges with positive residuals. By construction, only assignment edges contain labels from $S$, and each such edge $E_{ij}$ contains exactly three labels from $S$ (the labels in $X_j$). Thus, for each active assignment edge,
$$|U(E_{ij}) \cap S| \leq 3.$$

Residual sets in a covering are disjoint by definition ($U_i \cap U_j = \emptyset$ for $i \neq j$), so the sets $U(E_{ij}) \cap S$ corresponding to the $m$ active assignment edges are pairwise disjoint. Hence their union contains at most $3m$ labels from $S$.

On the other hand, $|S| = 3m$, and since only assignment edges contain labels from $S$ and $\mathcal{E}'$ is a covering, every label in $S$ must belong to one of these residual sets. Therefore the union of the $U(E_{ij}) \cap S$ has size exactly $3m$, which forces $|U(E_{ij}) \cap S| = 3$ for each active assignment edge and shows that their union is $S$.
\end{proof}

\begin{proof}[Proof of Theorem~\ref{thm:soundness}]
Suppose $\mathcal{E}' = (E'_1, \ldots, E'_k)$ is a covering with $F(\mathcal{E}') \leq C$.

By Lemma~\ref{lem:exactly-m}, there are exactly $m$ assignment edges with positive residuals. Enumerate them as $H_1, \ldots, H_m$.

By Lemma~\ref{lem:opening}, each $H_p$ is preceded by its opening edge in $\mathcal{E}'$.

By Lemma~\ref{lem:triplet}, the contributions of these $m$ assignment edges to $S$ are pairwise disjoint, each of size three, and their union is $S$.

For each $p \in [m]$, let $Q_p = U(H_p) \cap S$ be the three labels from $S$ contributed by $H_p$. By the construction of the reduction, there exists a triplet $X_p \in T$ such that
$$H_p = X_p \cup \{\omega_p, \tau_p\},$$
where $X_p$ is a three-element subset of $S$ with total weight $B$.

The only labels of $S$ contained in $H_p$ are those in $X_p$, so $Q_p$ and $X_p$ are both three-element subsets of $S$ contained in $H_p$. Since $Q_p$ consists precisely of the labels from $S$ first covered by $H_p$, we must have $Q_p = X_p$.

Because the sets $Q_1, \ldots, Q_m$ are pairwise disjoint and their union is $S$ (Lemma~\ref{lem:triplet}), the same holds for $X_1, \ldots, X_m$. Hence $\{X_1, \ldots, X_m\}$ forms a partition of $S$ into $m$ triplets, each of total weight $B$.

Since the elements of $S$ are labels for the elements of $A$, we translate back: for each $p$, the three labels in $X_p$ correspond to three elements of $A$ whose sum is $B$. The resulting sets $P_1, \ldots, P_m$ are pairwise disjoint and cover $A$, so $\mathcal{P} = \{P_1, \ldots, P_m\}$ is a valid 3-partition of $A$.
\end{proof}

\subsection{Conclusion of the Proof}

We have now established both directions:
\begin{enumerate}
\item \textbf{Theorem 3 (Completeness):} If $(A, B) \in \text{YES}_{3\text{-Partition}}$, then $\Phi(A, B) \in \text{YES}_{\text{OCP}}$.
\item \textbf{Theorem 4 (Soundness):} If $\Phi(A, B) \in \text{YES}_{\text{OCP}}$, then $(A, B) \in \text{YES}_{3\text{-Partition}}$.
\end{enumerate}

Moreover, the reduction $\Phi$ is computable in polynomial time, and the OCP instance has polynomial encoding size (weights are polynomially bounded; the budget is exponentially large in value but has polynomial bit-length).

\begin{proof}[Proof of Theorem 2]
By Subsection 3.1, OCP is in NP. By Subsections 3.2-3.4, there is a polynomial-time reduction from 3-Partition to OCP. Since 3-Partition is NP-complete \cite{garey1979computers}, it follows that OCP is NP-complete.
\end{proof}

\begin{remark}
This reduction does not establish strong NP-completeness of OCP. Determining whether OCP is strongly NP-complete remains an open question.
\end{remark}

\section{Examples}

In this section, we present two small instances (testA and testB) to illustrate the problem.
Figure~\ref{fig:graph_small} shows the hypergraphs associated with these small instances, where elements of $\mathcal{E}$ are represented as red circles and labels in $S$ are represented as blue squares. Each square has a label $s_i$ and a value $v_i$ (highlighted in yellow).

\textbf{Instance testA:}
\begin{align*}
S_A &= \{s_1, s_2, s_3, s_4, s_5, s_6, s_7\}, \\
\text{val}_A &= \{s_1 \mapsto 8, s_2 \mapsto 2, s_3 \mapsto 4, s_4 \mapsto 4, s_5 \mapsto 2, s_6 \mapsto 2, s_7 \mapsto 4\}, \\
\mathcal{E}_A &= \{E_1, E_2, E_3, E_4\}, \text{ where:} \\
E_1 &= \{s_1, s_2\}, \\
E_2 &= \{s_2, s_3, s_4, s_5\}, \\
E_3 &= \{s_4, s_5, s_6, s_7\}, \\
E_4 &= \{s_5, s_6\}.
\end{align*}

\textbf{Instance testB:}
\begin{align*}
S_B &= \{s_1, s_2, s_3, s_4, s_5\}, \\
\text{val}_B &= \{s_1 \mapsto 4, s_2 \mapsto 2, s_3 \mapsto 8, s_4 \mapsto 2, s_5 \mapsto 2\}, \\
\mathcal{E}_B &= \{E_1, E_2, E_3, E_4\}, \text{ where:} \\
E_1 &= \{s_1, s_2\}, \\
E_2 &= \{s_2, s_3, s_4, s_5\}, \\
E_3 &= \{s_3, s_4\}, \\
E_4 &= \{s_4, s_5\}.
\end{align*}

We analyzed these instances to determine the optimal ordered covering and compared it with a dynamic greedy heuristic. The greedy heuristic constructs the covering by iteratively selecting the subset $E \in \mathcal{E}$ that minimizes the residual weight (and thus the immediate cost).

For \textbf{testA}, the greedy algorithm successfully finds the optimal solution. The optimal sequence is $(E_4, E_3, E_2, E_1)$ with total cost 592.

However, the greedy strategy fails for \textbf{testB}. In this instance, the optimal sequence is $(E_4, E_3, E_2, E_1)$ with cost $2^{4} + 2^{8} + 2^{2} + 2^{4} = 292$.
The greedy algorithm, after selecting $E_4$ (cost $2^{4}$), selects $E_1$ (residual weight $6$, cost $2^{6}$) instead of $E_3$ (residual weight $8$, cost $2^{8}$), leading to a sequence $(E_4, E_1, E_2)$ with total cost $2^{4} + 2^{6} + 2^{2} + 2^{8} = 336$. The greedy choice of $E_1$ is locally cheaper than $E_3$, but it fails to reduce the cost of covering $s_2$ as effectively as the optimal sequence does.

These examples demonstrate that a locally optimal greedy strategy does not guarantee a globally optimal solution for the OCP, motivating the need for more sophisticated algorithms or exact approaches \cite{souza2024branch}.

\begin{figure}[h]
\centering
\includegraphics[width=0.8\textwidth]{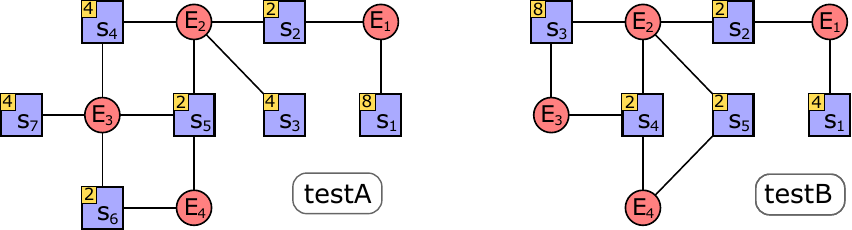}
\caption{Hypergraphs for small instances (testA and testB).}
\label{fig:graph_small}
\end{figure}

\section{Conclusion}

We have shown that the Ordered Covering Problem is NP-complete by giving a polynomial-time reduction from the strongly NP-complete 3-Partition problem. The construction uses opening and assignment edges together with a carefully chosen exponential cost separation to enforce that any feasible solution has exactly $m$ active assignment edges, each opened in advance and covering one triplet of labels, yielding a bijection between feasible coverings and valid 3-partitions.

This result complements existing NP-completeness results for vertex ordering problems in distance geometry and explains why heuristic and exact exponential-time methods are unavoidable for OCP. In particular, it provides a theoretical justification for the greedy ordering heuristic and Branch-and-Bound approaches developed for the SBBU algorithm, and it motivates further work on approximation guarantees, parameterized algorithms, and structural properties of practically arising OCP instances.

\section*{Acknowledgements}

We thank the Brazilian research agencies CNPq and FAPESP for financial support. We also thank the anonymous reviewers for their valuable comments and suggestions that improved the presentation of this paper.

\bibliographystyle{plain}
\bibliography{references}

\end{document}